\documentclass[sigconf,nonacm]{aamas}

\usepackage{algorithm}
\usepackage{algorithmic}
\usepackage{amsmath}
\usepackage{enumitem}
\usepackage{etoolbox}
\usepackage{bm}
\usepackage{tikz}
\usepackage{multirow}
\usepackage{framed}
\usepackage{mathrsfs}
\usepackage{subfigure}

\usepackage{pifont}       
\usepackage{bbding}       
\usepackage{fontawesome}  
 
\usepackage{balance} 
\newtheorem{definition}{Definition}
\newtheorem{theorem}{Theorem}

\newtheorem{proposition}{Proposition}

\usepackage{color-edits}
\addauthor[Yuhan]{yuhan}{blue}
\addauthor[Tian]{tian}{blue}
\addauthor[Junyu]{junyu}{blue}

\title{Combinatorial Diffusion Auction Design} 

\author{Xuanyu Li}
\affiliation{
  \institution{Key Laboratory of Intelligent
Perception and Human-Machine
Collaboration, ShanghaiTech
University}
  \city{Shanghai}
  \country{China}}
\email{lixy22022@shanghaitech.edu.cn}

\author{Miao Li}
\affiliation{
  \institution{Key Laboratory of Intelligent
Perception and Human-Machine
Collaboration, ShanghaiTech
University}
  \city{Shanghai}
  \country{China}}
\email{limiao@shanghaitech.edu.cn}

\author{Yuhan Cao}
\affiliation{
  \institution{Key Laboratory of Intelligent
Perception and Human-Machine
Collaboration, ShanghaiTech
University}
  \city{Shanghai}
  \country{China}}
\email{caoyh1@shanghaitech.edu.cn}

\author{Dengji Zhao}
\affiliation{
  \institution{Key Laboratory of Intelligent
Perception and Human-Machine
Collaboration, ShanghaiTech
University}
  \city{Shanghai}
  \country{China}}
\email{zhaodj@shanghaitech.edu.cn}

\begin{abstract}
Diffusion auction design for combinatorial settings is a long-standing challenge. One difficulty is that we cannot directly extend the solutions for simpler settings to combinatorial settings (like extending the Vickrey auction to VCG in the traditional settings). In this paper, we propose a different approach to leverage the diffusion auctions for single-item settings. We design a combinatorial diffusion auction framework which can use any desirable single-item diffusion auction to produce a combinatorial auction to satisfy incentive compatibility (IC), individual rationality (IR), and weak budget balance (WBB). 

\end{abstract}


\newcommand{\BibTeX}{\rm B\kern-.05em{\sc i\kern-.025em b}\kern-.08em\TeX}

\begin{document}


\pagestyle{fancy}
\fancyhead{}


\maketitle 

%
%
%
\section{Introduction}




Auction theory plays a significant role in mechanism design, which integrates economics and artificial intelligence \cite{nisan2009algorithmic}. Combinatorial auctions represent a complex auction setting. They allow bidders to place bids on combinations of items rather than on individual items. These auctions are particularly valuable when items have complementary or substitutive values, meaning the value of a combination may differ from the sum of the values of individual item. Traditional applications of combinatorial auctions include spectrum sales, logistics, and procurement, where complex interdependencies between items are common.


Diffusion auction design is a new trend in mechanism design which considers the connections between buyers and asks buyers to report them~\cite{li2017mechanism,li2019diffusion}.
The main purpose of a diffusion auction is to involve more participants to the auction so that we could improve social welfare or the seller's revenue.
The literature has proposed many interesting mechanisms for selling a single-item including multi-unit case~\cite{GuoH21}.
However, none of them can be easily extended to the combinatorial setting.
The difficulty in a nutshell is that inviting more neighbors in a combinatorial setting can sometimes be mimicked by reporting a combined valuation report of a buyer and his neighbors. This is not possible in simpler settings as the buyers' valuation reports are limited, e.g., in a single-item setting, a buyer can only report one value.



In this paper, we do not try to directly extend the principle of a single-item diffusion auction to a combinatorial setting, which was proved seems impossible~\cite{zhao2019selling,kawasaki2019strategyproof}.
Instead, we proposed a combinatorial framework to transfer a combinatorial problem into multiple single-item/bundle sub-problems in multiple stages and apply single-item diffusion auctions to produce the final outcome.

Given any desirable single-item diffusion auction, the framework will produce a combinatorial auction to satisfy incentive compatibility (IC), individual rationality (IR), and weak budget balance (WBB). 
Additionally, we study a specific mechanism generated by framework and show the effectiveness of the framework. Finally, we show that the generated mechanisms can provide positive incentives to non-winners, an important but often overlooked issue in diffusion auctions.

In summary, our work advances the state of the art as the following:
\begin{itemize} 
\item We introduce the diffusion combinatorial auction framework (DCAF) on social networks, in which we can use classic single-item diffusion mechanisms with IC and IR to design combinatorial diffusion mechanisms with IC, IR, and WBB. 

\item We propose the property of existing positive incentive for non-winner (EPI4NW) which says that for the bidders whose allocation is empty, they can still get positive utility. We prove the mechanisms produced by DCAF satisfy EPI4NW. 

\item We then give a specific instance based on our framework named dealer retail mechanism (DRM). It is the first combinatorial diffusion mechanism satisfying EPI4NW, IC, IR, and WBB.
\end{itemize}
\section{Related Work}
\subsection{Combinatorial Auction}
Combinatorial auctions are a fundamental aspect of mechanism design, enabling the simultaneous auctioning of multiple items while allowing bidders to submit bids on combinations of items. The Vickrey-Clarke-Groves (VCG) mechanism~\cite{vickrey1961counterspeculation,clarke1971multipart,groves1973incentives} is widely regarded as a benchmark for combinatorial auctions, ensuring incentive compatibility (IC), individual rationality (IR), and the maximization of social welfare. However, the VCG mechanism is hindered by its computational complexity, notably the NP-hardness of determining the winners~\cite{assadi2020improved}.

To address the computational challenge, several efficient mechanisms have been proposed. Notably, \citet{dobzinski2006truthful} introduced the first IC and computationally efficient mechanism for general monotone valuations. Further developments include mechanisms tailored for sub-additive valuations \cite{dobzinski2007two} and XOS valuations \cite{assadi2020improved}, each offered improved approximation ratios on social welfare.

\subsection{Diffusion Auction}
In 1994, a pioneering work of \citet{bulow1994auctions} proved that one extra bidder in the simplest second price auction outperform the optimal auction without the extra bidder in terms of  seller's revenue under the assumption of symmetric type. This finding shifted academic focus towards attracting more participants to auctions, rather than solely pursing the design of optimal auction.

\citet{li2017mechanism} was the first to study auctions under the setting of social network, where buyers are incentivized to invite their neighbors to participate in the auction. This work has inspired numerous extensions~\cite{zhao2019selling,kawasaki2019strategyproof,fang2023multiunit}, and similar studies have emerged in other domains, such as matching~\cite{you2022strategy,Yang2022OneSidedMW} and cooperative games~\cite{zhang2022incentives}.

Multi-unit auctions in social networks address scenarios where a seller offers multiple homogeneous items, and buyers can strategically report their valuations and social connections. These auctions aim to maximize seller revenue by attracting more buyers. Previous mechanisms, such as GIDM~\cite{zhao2019selling} and DNA-MU~\cite{kawasaki2019strategyproof}, failed to ensure IC due to potential buyer manipulations. SNCA~\cite{Xiao_Song_Khoussainov_2022} introduced budget constraints but deviated from standard multi-unit diffusion auctions.

Recent efforts have introduced two mechanisms that satisfy truthfulness in multi-unit diffusion auctions.
The LDM-Tree mechanism~\cite{DBLP:conf/atal/LiuLZ23} localizes competition within layers of a tree based on the distance of agents from the seller, while MUDAN~\cite{fang2023multiunit} uses an iterative approach to allocate items during graph exploration. Our mechanism offers a novel approach to incentivizing buyers by providing profit, akin to IDM, but addressing the greater complexity inherent in combinatorial auction scenarios.

There are currently few works on combinatorial auctions on social networks. MetaMSN-m \cite{fang2024meta} provides a combinatorial diffusion auction framework, but it cannot generate mechanisms satisfy EPI4NW.

\section{The Model}

We consider the problem of an auction in the social network. A seller $s$ wants to sell $m$ items $M = \{1,\cdots,m\}$ in a digraph $G=(N,E)$. There are $n$ potential bidders denoted by $N=\{1,\cdots,n\}$ in the social network. Edge $(i,j)\in E$ indicates that bidders $i$ and $j$ are neighbors of each other. Let $B=2^M$ denote the total set of all possible item bundles. Each bidder $i$ is interested in all different bundles of items. Each bundle can be represented as $b \in B$. Each bidder $i\in N$ has different valuations to every bundle, which is defined as a valuation function $v_{i}(b)\in \mathbb{R}^{+}$. $v_{i}(b)$ means the cost that bidder $i$ is willing to pay for the bundle $b\in B$. According to the normal situation, we assume the valuation function $v_{i}$ is monotone, i.e., $v_{i}(x)\leq v_{i}(y)$ for all bundle $x,y\in B, x\subseteq y$. 

 We consider that the seller and all bidders are from a social network. Each bidder $i \in N$ has a set of neighbor bidders denoted by $r_{i}\subseteq N\cup\{s\}$ who are directly connected to her in the social network and she does not know anyone else except her neighbors. In particular, $r_s$ represents the neighbor set of the seller $s$. 

 At the beginning, we assume only the seller $s$ and her neighbors $r_s$ know the auction. In order to attract more bidders to join, the seller needs her neighbors to invite their neighbors to the auction.  All participants of the auction can decide whether to spreads the auction information to their neighbors. The number of participants increases until no new bidders are invited. 
 
Let $\theta_i=(v_i,r_i)$ be the type of bidder $i\in N$ and $\theta = (\theta_{1},\cdots,\theta_{n})$ as the type profile of all bidders. The type $\theta_i$ is private information of bidder $i$. She may misreport her true information for more benefits in some situations. Let $\theta^\prime_i = (v^\prime_i, r^\prime_i)$ be the reported profile of bidder $i$, where $v^\prime_i: B\rightarrow \mathbb{R^{+}}$ and $r^\prime_i \subseteq r_i$, meaning that the reported neighbor set is a subset of the true neighbor set. In this case, $\theta^\prime = (\theta^\prime_1, \cdots, \theta^\prime_n)$ is defined as the report profile of all bidders in $N$. For convenience, we use $\theta_{-i}$ to represent the reported profile of all bidders except for bidder $i$, and $\theta^\prime$ can also be written as $(\theta_i, \theta_{-i})$. Let $T$ be a bidder set, the reported profile of bidders $i\in T$ is represented as $\theta^\prime_T=\bigcup_{i\in T} \theta_i^\prime$. We use $\Theta$ to represent the space of $\theta$.

Given a reported profile $\theta^\prime$, the seller $s$ can use a mechanism $\mathcal{M}$ to determine how to allocate the items and how much each bidder needs to pay.

\begin{definition}[Auction Mechanism]
    An auction mechanism $\mathcal{M}$ consists of an allocation policy $\pi = (\pi_{i})_{i\in N}$ and a payment policy $p=(p_{i})_{i\in N}$, where $\pi_{i} \in B$ and $p_{i}\in \mathbb{R}$ are allocation function and payment function for bidder $i$ respectively.
\end{definition}

In the real scenario, only properly invited bidders can join in the auction. Hence, we use the type reported by bidders to identify whether the bidders are qualified. the bidder $i$ is qualified under reported type if and only if there is at least one path between seller $s$ and $i$. This means there is a bidder sequence $(j_1,j_2,\cdots,j_l,i)$, where $j_1\in r_s$ and for all $1<t\leq l$, $j_t\in r_{j_{t-1}}$, $i\in j_l$. Each bidder can only appear once in the sequence. 

Let $Q(\theta^\prime)$ be the set of all qualified agents under $\theta^\prime$. A special type of auctions, diffusion auction, works on only qualified bidders. 

\begin{definition}[Diffusion Auction Mechanism]
A diffusion auction mechanism is an auction mechanism such that for all reported type profiles $\theta^\prime$, it satisfies:
\begin{enumerate}
    \item For any unqualified agent $i\notin Q(\theta^\prime)$,  $\pi_i(\theta^\prime) = \emptyset$ and $p_i(\theta^\prime) = 0$. 
    \item For any qualified agent $i\in Q(\theta^\prime)$, $\pi_{i}(\theta^\prime)$ and $p_{i}(\theta^\prime)$ are independent of the types of all unqualified agents. 
\end{enumerate}
\end{definition} 

Given a bidder $i$ of type $\theta_{i}=(v_{i},r_{i})$ and a report profile $\theta^\prime$, the utility of $i$ under mechanism $\mathcal{M}=(\pi,p)$ is defined as $u_{i}(\theta_{i},\theta^\prime,(\pi,p))=v_{i}(\pi_{i}(\theta^\prime))-p_{i}(\theta^\prime)$. We use $u_{i}(\theta^\prime)$ to represent $u_{i}(\theta_{i},\theta^\prime,(\pi,p))$ for simplicity.

An ideal mechanism $\mathcal{M}$ should satisfy individual rationality, incentive compatibility, and weak budget balance.

We say a diffusion auction mechanism satisfies individually rationality if for each bidder, her utility is non-negative when she truthfully reports her valuation, no matter how many neighbors she invites and what the other bidders do. It means that invitation action will not make bidders' utility negative as long as they report valuation truthfully. 

\begin{definition}[Individual Rationality]
  A diffusion auction mechanism $\mathcal{M}=(\pi,p)$ satisfies individual rationality (IR) if for all $i\in N$, all $r_{i}^\prime\subseteq r_{i}$, and all $\theta_{-i}^\prime$, $u_{i}(\theta_{i}^\prime=(v_{i},r_{i}^\prime),\theta_{-i}^\prime,(\pi,p))\geq 0$.
\end{definition}

Apart from IR, another important property satisfies incentive compatibility. IC means that for each bidder $i$, reporting her valuation and fully broadcasting the auction information to her neighbors is always a dominant strategy, no matter what the others do. It guarantees that each bidder will report her valuation for each bundle truthfully and invite all her neighbours to join the auction.

\begin{definition}[Incentive Compatibility]
  A diffusion auction mechanism $\mathcal{M}=(\pi,p)$ satisfies incentive compatibility (IC) if for all $i\in N$, all $\theta^\prime$, and $\theta^\prime_{-i}$, $u_{i}(\theta_{i},\theta_{-i}^\prime,(\pi,p))\geq u_{i}(\theta_{i}^\prime,\theta_{-i}^\prime,(\pi,p))$.
\end{definition}

The following property is related to the seller's revenue from the auction. The seller's revenue equals to the sum of all bidders' payments. 
\begin{definition}[Seller's Revenue]
    Given a report profile $\theta^\prime$ and a diffusion auction mechanism $\mathcal{M}=(\pi,p)$, the revenue or profit of the seller is defined by the sum of all bidders’ payments, $\mathcal{R}(\theta^\prime)=\sum_{i\in {N}}{p_i(\theta^\prime)}$.
\end{definition}

We say that a mechanism $\mathcal{M}$ satisfies weakly budget balance if the seller's revenue is always non-negative from the auction. WBB ensures that seller will not lose in the auction and WBB is also an important motivation for sellers.   

\begin{definition}[Weak Budget Balance]
  A diffusion auction mechanism $\mathcal{M}=(\pi,p)$ satisfies weak budget balance (WBB) if for all $\theta^\prime\in \Theta$, $ \mathcal{R}(\theta^\prime)=\sum_{i\in {N}}{p_i(\theta^\prime)}\geq 0$.
\end{definition}

Our goal is to design diffusion combinatorial auction mechanisms which are IC, IR, and WBB.

\section{Diffusion Combinatorial Auction Framework}

In this section, we introduce a framework that can generate diffusion combinatorial auction mechanisms, based on existing single-item diffusion auction mechanisms. 

To the best of our knowledge, there is no such work considering the general combinatorial auction in a diffusion setting. MUDAN \cite{fang2023multiunit} is a mechanism that only considers multi-unit auctions in a social network, where inviting actions can increase the bidders’ chances of obtaining the items. However, the motivation is not strong enough comparing with directly getting profit (negative payments) when the bidder wins no item. Therefore, we are committed to solving this difficult problem. Another multi-unit mechanism, named LDM-Tree \cite{DBLP:conf/atal/LiuLZ23}, can allocate profit to those who win no items. However, it employs competition localization, which means only a few (several layers) bidders can earn the profit. Therefore, enabling more bidders to gain profit is a key consideration in mechanism design.

The core thought of our framework is based on single-item diffusion auction mechanisms. Two questions arise before proceeding to a single-item diffusion auction: which distributors are allowed to conduct local auctions in what range, and what items to sell for each distributor separately. Thus, we divide the whole framework into three parts. Each part can be described as a class of processes. The first part chooses candidate distributors for local auctions. The second part However, only when each of the processes satisfies some specific properties, the combination of them with our framework can produce a complete mechanism satisfying IC, IR, and WBB. \ref{fig:enter-label} shows the complete running process of the mechanisms designed by the framework. 

\begin{figure}
    \centering
    \includegraphics[width=1\linewidth]{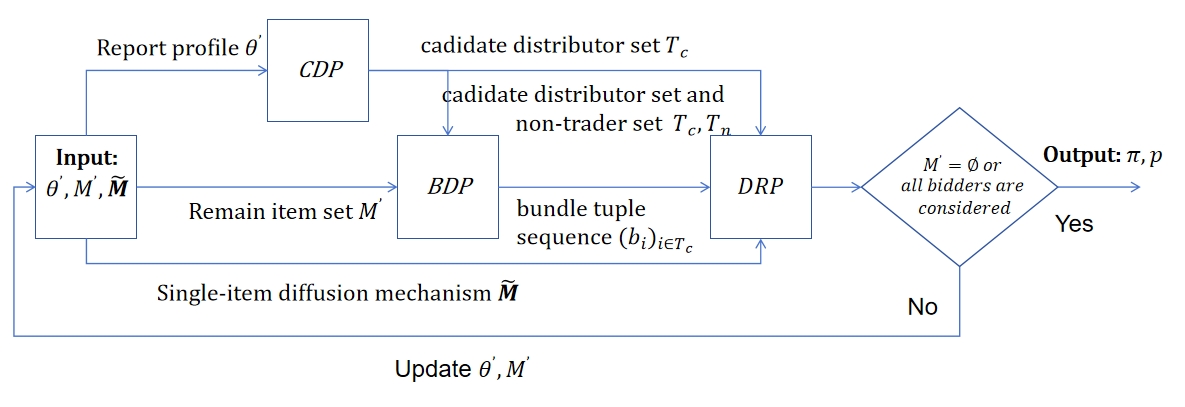}
    \caption{The Pipeline of Diffusion Combinatorial Auction Framework}
    \label{fig:enter-label}
\end{figure}

\subsection{Candidate Distributor Determination Process}

In this section we consider the first question: how to choose distributors. We define the process of choosing the candidate distributors as a function, which only takes the neighbors' reports as parameters and is independent to the valuations. 

\begin{definition}
    [Candidate Distributor Determination Process]
         We define $CDP: r^\prime \mapsto (T_c,T_n)$ as a process of screening out a set of candidate distributors $T_c$ and a set of non-trading bidders $T_n$, where $r^\prime=(r_i^\prime)_{i \in N \cup \{s\}}$ denotes all neighbor information reported. 
\end{definition}

A trivial approach of CDP is that $T_c=\{s\}$ and $T_n=\emptyset$. Then the only candidate distributor is the seller $s$. Additionally,
 we give another example of candidate distributor determination process. The key idea is inspired by the works \cite{fang2023multiunit,Li2024DoubleAO}. They use the graph exploration to utilize network structure. 
\begin{framed}
{\noindent\textbf{Graph Exploration CDP:}} \\
\noindent\rule{\textwidth}{0.5pt}
\textbf{Input:} $\theta^\prime$
\begin{enumerate}
 \item Index seller $s$'s neighbors as $(1,2,\cdots,|r_s^\prime|)$ where $|r_i^\prime|\geq |r_j^\prime|$ for any bidder pair $i,j\in r_s$ and $i<j$. Partition former $\lceil{\frac{1}{2}|r_s|}\rceil$ bidders into $T_c$ and others into $T_n$.
 \item Loop: 
   \begin{enumerate}
        \item Calculate the neighbors of non-trading bidders $$R(T_n)=\bigcup_{i\in T_n} r_i^\prime.$$ Index the neighbors as $(1,2,\cdots,|R(T_n)|)$, where $|r_i^\prime|\geq |r_j^\prime|$ for $i<j$.
        \item Partition former $\lceil{\frac{1}{2}|R(T_n)|}\rceil$ bidders into $T_c$ and other bidders into $T_n$.
        \item If $R(T_n) \subseteq (T_c \cup T_n)$, terminate the loop.
    \end{enumerate}
\end{enumerate}
\textbf{Output:} 
Candidate distributor set $T_c$ and non-trading set $T_n$.
\end{framed}
We can use the non-trading bidders' reported profile $\theta^\prime_{T_n}$ to decide the price and the fixed resale revenue of each bundle $b$. 

\begin{definition}[Pricing Function]   
         For each bundle $b\in B$, the price of $b$ is determined by a pricing function $Pr(\theta^\prime_{T_n},b)$.
\end{definition}

\begin{definition}[Resale Revenue Function]   
    For each bundle $b\in B$, the revenue of $b$ in the resale process is determined by the resale revenue function $Rev(\theta^\prime_{T_n},b) > Pr(\theta^\prime_{T_n},b)$.  
\end{definition}
Resale revenue function calculates the fixed resale revenue for each bundle. For all bidder $i\in T_c$, the resale revenue of $i$ is fixed for same bundle $b$. For example, pricing function outputs the second price of bidders $T_n$ and resale revenue function outputs the first price of bidders $T_n$. 
\subsection{Bundle Division Process}

Next, we consider the second problem before holding a local single-item diffusion auction. Choose reasonable bundles of items for each candidate distributor $i\in T_c$. We give each candidate distributors $i\in T_c$ a tuple of bundles $b_i=(b_i^*, b_i^\circ)$. $b_i^*$ is the potential resale bundle for candidate distributor $i$ and $b_i^\circ$ is the reserve bundle for candidate distributor $i$. Intuitively, $BDP$ gives the candidate distributors two chances to get positive utility. The candidate distributor $i$ firstly tries to resell bundle $b_i^*$. When $i$ cannot resell the bundle $b_i^*$, $i$ 
 will consider to reserve bundle $b_i^\circ$. With reserve bundle $b_i^\circ$, $i$ can try to sell a bundle $b_i^*$ which is more valuable to other bidders but less valuable to herself, thereby avoiding the risk of missing a valuable bundle.

We give the definition of bundle choosing process as: 
\begin{definition}
    [Bundle Division Process]
         We define $BDP: (\theta^\prime,\hat{M},T_c,T_n,Pr,Rev)\mapsto(b_1,b_2,...,b_{|T_c|})$ as a process of division of remain item bundle $\hat{M}$ for each candidate distributor $i\in T_c$, which satisfies
         \begin{itemize}
             \item $\bigcup_{i\in T_c} \{b_i^*\cup b_i^\circ\}\subseteq \hat{M}$.
             \item For every two candidate distributors $i,j\in T_c$, their bundle $(b_i^*\cup b_i^\circ)\cap (b_j^*\cup b_j^\circ)=\emptyset$.
         \end{itemize}
        
\end{definition}

The first property ensures the items in each bundle are valid to be resold and the second property guarantees each item can only be sold once.




Then we give an example of bundle division process.   
A trivial approach is to randomly choose one item for each candidate distributor in $T_c$. 
We will show another natural approach inspired by greedy idea in 5.2.

\begin{framed}
{\noindent\textbf{Random Single Item BDP:}} \\
 \noindent\rule{\textwidth}{0.5pt} 
 \label{BDP:random}
\textbf{Input:} $\theta^\prime$, $\hat{M}$, $T_c$, $T_n$, $Pr$, $Rev$.

\begin{enumerate}
    \item Randomly order the bidders in $T_c$ as $(1,2,\cdots,|T_c|)$.  
 \item Sequentially for $1,2,\cdots,|T_c|$: random choose an item $m$ from $\hat{M}$. Let $b_i=(b_i^*=\{m\},b_i^\circ=\{m\})$. Update $\hat{M}=\hat{M}\setminus \{m\}$.
\end{enumerate}
\textbf{Output:} 
The bundle tuple sequence $(b_1,b_2,\cdots,b_{|T_c|})$ divided for each bidder $i\in T_c$.
\end{framed}

\subsection{Diffusion Resale Process}

The last and the most significant section of our diffusion combinatorial auction framework is the diffusion resell process. In this process, we consider the candidate distributors as local sellers with specific bundles given by BDP and let them hold a single item diffusion auction among their critical children and try to sell potential resale bundles.

Firstly, we introduce some definitions based on the graph structure. Considering the social network as an unweighted graph, for any determined node $i\in V$, there would be multiple simple paths from seller $s$ to $i$. Additionally, there are several cut nodes shared by all these paths. These nodes are termed critical diffusion nodes for bidder $i$, as the absence of any of these nodes would prevent bidder $i$ from receiving the sale information and participating in the auction.

\begin{definition}
[Critical Diffusion Nodes] Given a reported type profile $\theta^\prime$ and a node $i \in V \setminus \{s\}$, we define $C_{i}\left(\theta^\prime\right)=\bigcap_{L \in \mathbb{L}_{i}(\theta^\prime)} L$ as the set of $i$'s critical diffusion nodes where $\mathbb{L}_{i}\left(\theta^\prime\right)$ is the set of all feasible paths from seller $s$ to node $i$ with $\theta^\prime$. Note that $i \in C_i(\theta^\prime)$. In short, critical diffusion nodes are the cut nodes shared by all feasible paths from $s$ to $i$. 
\end{definition}

Based on the definition of critical diffusion nodes, we define the critical diffusion sequence of bidder $i$ as an unique fully order set.

\begin{definition}
    [Critical Diffusion Sequence] The critical diffusion sequence of bidder $i$ is $C_{i}^{*}\left(\theta^\prime\right)=\left\{ds_{1}, \cdots, ds_{k}, ds_{k+1}, \cdots, i\right\}$ for all nodes in $C_i(\theta^\prime)$, where $ds_k \in C_{ds_{k+1}}(\theta^\prime)$. In this sequence, the former term is the critical diffusion nodes of latter term. 
\end{definition}

Then we give another relationship named critical children. Bidder $i$'s critical children are the bidders cut by bidder $i$ in the social network. 

\begin{definition}
    [Critical Children]
         We define  $D_{i}\left(\theta^\prime\right)$ as the set of $i$'s critical children where for all bidder $j \in D_{i}\left(\theta^\prime\right)$, $i$ is $j$'s critical diffusion node: $i \in C_j(\theta^\prime)$. Specially, $i \in D_i(\theta^\prime)$.
\end{definition}

The most important process of $DCAF$ is the diffusion resale process (DRP). We can consider $DRP$ as a single-item auction held by candidate distributor $i$ in her critical children.  

\begin{definition}
    [Diffusion Resale Process]
         We define $DRP: (\theta^\prime,T_c,(b_i)_{i\in T_c},Pr,Rev,\mathcal{\tilde{M}})\mapsto (\pi^*, p^*)$ as a diffusion resale process, where $\pi^*=\pi_i(\theta^\prime)$, $p^*=p_i(\theta^\prime)$, for all bidder $i\in \bigcup_{j\in T_c} D_j(\theta^\prime)$.The single-item diffusion auction mechanism in DRP process should satisfy revenue consistent property.
\end{definition}

We give out detailed allocation and payment of DRP in different situations below.

\begin{framed}
{\noindent\textbf{Diffusion Resale Process (DRP):}} \\
 \noindent\rule{\textwidth}{0.5pt}
 \textbf{Input:} $\theta^\prime$, $T_c$, $(b_i)_{i\in T_c}$, $Pr$, $Rev$, $\mathcal{\tilde{M}}$
 
 For each candidate distributor $i\in T_c$:
\begin{enumerate} 
    \item  Consider the whole bundle $b_i$ as a single item and apply $\mathcal{\tilde{M}}$ in bidders $D_i(\theta^\prime)$ to get $\tilde{\pi_j}$ and $\tilde{p_j}$ for all bidder $j\in D_i(\theta^\prime)$.
    \item If $\mathcal{\tilde{R}}(\theta^\prime)\geq Rev(\theta^\prime, b_i^*)$, then 
   \begin{equation*}
    \pi_j(\theta^\prime)=
        \begin{cases}
        \tilde{\pi}_j(\theta^\prime), & j\in D_i(\theta^\prime)\setminus\{i\}\\
        \emptyset, & j=i \\
        \end{cases}
    \end{equation*} 
   \begin{equation*}
    p_j(\theta^\prime)=
        \begin{cases}
        \tilde{p}_j(\theta^\prime), \quad \quad \quad \quad j\in D_i(\theta^\prime)\setminus\{i\}\\
        Pr(\theta', b_i^*)-Rev(\theta', b_i^*), \quad j=i \\
        \end{cases}
    \end{equation*} 
    \item Otherwise, 
    \begin{equation*}
    \pi_j(\theta^\prime)=
        \begin{cases}
        \emptyset, & j\in D_i(\theta^\prime)\setminus\{i\}\\
        b_i^\circ, & j=i \\
        \end{cases}
    \end{equation*} 
   \begin{equation*}
    p_j(\theta^\prime)=
        \begin{cases}
        0, & j\in D_i(\theta^\prime)\setminus\{i\}\\
        Pr(\theta', b_i^\circ), & j=i \\
        \end{cases}
    \end{equation*} 
\end{enumerate}
  \textbf{Output:} $\pi_j^*$ and $p_j^*$ for all bidder $j\in D_i(\theta^\prime)$. 
\end{framed}

\subsection{Diffusion Combinatorial Auction Framework}
Finally, we propose the whole diffusion combinatorial auction framework (DCAF). This framework involves all previous processes and it provides the allocation and payment results. The formal description of $DCAF$ is shown as follows. 
\begin{framed}
{\noindent\textbf{Diffusion Combinatorial Auction Framework (DCAF):}} \\
 \noindent\rule{\textwidth}{0.5pt}
\textbf{Input:} $\theta^\prime$, $\mathcal{\tilde{M}}$, $Pr$, $Rev$

\begin{enumerate}
\item Initialize the remain item set $\hat{M}=M$, the unconsidered participant set $\hat{N}=N\cup \{s\}$ 

    \item Use profile $\theta^\prime$ as the input of a $CDP$, get $CDP(r^\prime)=(T_c,T_n)$, where $T_c\cup T_n\subseteq \hat{N}$. 
    \item Use a $BDP(\theta^\prime,\hat{M},T_c,T_n,Pr,Rev)$ to find resale bundles $b_i$ for each candidate distributor $i\in T_c$. 
    \item For each candidate distributor $i\in T_c$, calculate $(\pi_j,p_j)=DRP(\theta^\prime,T_c,(b_i)_{i\in T_c},Pr,Rev,\mathcal{\tilde{M}})$ for each bidder $j\in D_i(\theta^\prime)$.
    \item For each bidder $i\in T_n$, her allocation is $\pi_i=\emptyset$ and payment is $p_i=0$.
     \item Update $\hat{N}=\hat{N}\setminus \{\bigcup_{i\in T_c} D_i(\theta^\prime)\cup T_n\}$ and  $\hat{M}=\hat{M}\setminus\{\bigcup_{j\notin \hat{N}} \pi_j(\theta)\}$.
    \item Then we get a new sub-problem with remain item set $\hat{M}$ unconsidered participant set $\hat{N}$ and $(\theta^\prime_i)_{i \in \hat{N}}$. \item Repeat step 2 to step 7 until $\hat{M}=\emptyset$ or $\hat{N}=s$.
    \item For each bidder $i\in \hat{N}\setminus s$, her allocation is $\pi_i=\emptyset$ and payment is $p_i=0$.
\end{enumerate}
\textbf{Output:} 
Allocation $\pi$ and payment $p$.
\end{framed}



 

To ensure we can eventually generate an IC mechanism, there are three sufficient conditions:
\begin{itemize}
    \item The $CDP$ satisfies \emph{candidate distributor consistency}.
    \item The $BDP$ satisfies \emph{resale diffusion monotonicity}.
    \item The $DRP$ satisfies \emph{revenue consistency}.
\end{itemize}

\begin{definition}[Candidate Distributor Consistency]
Assume there are two profiles $r^{\prime_1}$, $r^{\prime_2}$. Let $CDP(r^{\prime_1})=(T_c^1,T_n^1)$ and $CDP(r^{\prime_2})=(T_c^2,T_n^2)$. We say a $CDP$ satisfies candidate distributor consistency:
\begin{itemize}
\item  If $i\in T_c^1$, $r_i^{\prime_1}\subseteq r_i^{\prime_2}$ and $r_j^{\prime_1}=r_j^{\prime_2}$ for $j\neq i$, then $i\in T_c^2$.
\item  If $i\in T_n^1$, $r_i^{\prime_1}=r_i$ and $r_j^{\prime_1}=r_j^{\prime_2}$ for $j\neq i$, $i\in T_n^2$.
\item  If $i\in N\setminus (T_n^1\cup T_c^1)$, $r_i^{\prime_1}\subseteq r_i^{\prime_2}$ and $r_j^{\prime_1}=r_j^{\prime_2}$ for $j\neq i$, then $CDP(r^{\prime_1})=CDP(r^{\prime_2})$ .
\end{itemize}
\end{definition}

The first property means when other bidders' reported neighbors not change, the candidate distributor who reports more neighbors will not change her candidacy. The second property shows that bidders once belongs non-trading set when bidder reports all neighbors, they can not leave the non-trading set. Last property means the bidders who are neither in candidate distributor set
nor in non-trading set can not change the solution of $CDP$ when other bidders' reported neighbors are fixed.

\begin{definition}[Resale Diffusion Monotonicity]
    Assume there are two profiles $\theta^{\prime_1}$, $\theta^{\prime_2}$. Let $BDP(\theta^{\prime_1},\hat{M},T_c,T_n,Pr,Rev)=$ \\ $(b_1^1,b_2^1,\cdots,b_{|T_c|}^1)$ and $BDP(\theta^{\prime_2},\hat{M},T_c,T_n,Pr,Rev)=(b_1^2,b_2^2,\cdots,b_{|T_c|}^2)$. For all candidate distributor $i\in T_c$, if $r_i^{\prime1}\subseteq r_i^{\prime2}$, $u_i(\theta^{\prime_1})\le u_i(\theta^{\prime_2})$.
\end{definition}

Resale diffusion monotonicity means that for any candidate distributor $i$, her utility will increase when her reported neighbor scale extends.

\begin{definition}[Revenue Consistency]
    Given a report profile $\theta^\prime$, we define the revenue of seller $\tilde{s}$ in a single-item diffusion auction mechanism $\mathcal{\tilde{M}}=(\tilde{\pi},\tilde{p})$ as $\tilde{\mathcal{R}}(\theta^\prime)=\sum_{i\in D_{\tilde s}(\theta^\prime)}{p_i(\theta^\prime)}$. If for each bidder $i\in D_{\tilde s}$ and all possible $\theta^\prime\in\Theta$, if $u_{i}(\theta_{i},\theta_{-i},(\tilde{\pi},\tilde{p}))>0$ and $\tilde{\mathcal{R}}(\theta)<Rev(\theta, b)$. When she misreports $\theta_i$ as $\theta_i^\prime$, and $u_{i}(\theta_{i}^\prime,\theta_{-i}^\prime,(\tilde{\pi},\tilde{p}))>0$, we can ensure $\tilde{\mathcal{R}}(\theta^\prime)<Rev(\theta^\prime, b)$, for all bundle $b\in B$ then we say $\mathcal{\tilde{M}}$ is revenue consistent.
\end{definition}

A mechanism $\mathcal{\tilde{M}}$ satisfies revenue consistency when all  seller's critical children $i\in D_s$ ,whose utility is positive, can not influence the relationship between seller's revenue of selling bundle $b$ and fixed revenue of bundle $b$ under $\theta^\prime$
 via misreporting.

\begin{theorem}
    Suppose that a single-item mechanism $\mathcal{\tilde{M}}$ satisfies IC, IR, and revenue consistency, a $CDP$ satisfies candidate distributor consistency, a $BDP$ satisfies resale diffusion monotonicity. Then the corresponding mechanisms generated by applying $DCAF$ are IC, IR, and WBB.
\end{theorem}

\begin{proof}
    First, we consider IR for each round of step 2 to step 6. 
    
    \textbf{Case 1: $i\in T_n$}. For bidder $i$, her allocation $\pi_i(\theta^\prime)=\emptyset$ and her payment $p_i=0$.
        
    \textbf{Case 2: $i\in T_c$}. For bidder $i$, when she resells bundle through $DRP$, her allocation $\pi_i(\theta^\prime)=\emptyset$ and her payment $p_i\geq 0$.
    When she reserves bundle through $DRP$, her utility  $u_i(\theta^\prime)=v_i^\prime(b_i)-p_i(\theta^\prime)\geq 0$.
    Else, her allocation $\pi_i(\theta^\prime)=\emptyset$ and her payment $p_i=0$.

    \textbf{Case 3: $i\notin T_n$}.
    Since $\mathcal{\tilde{M}}$ is IR, for all bidder $i\in \bigcup_{j\in T_c} D_j(\theta^\prime)$, we can get $u_i(\theta^\prime)\geq 0$.   
    For other bidders $i\notin T_n$, we can get $\pi_i(\theta^\prime)=\emptyset$ and $p_i(\theta^\prime)=0$.

  It is concluded that no matter which category bidder $i$ is classified by $CDP$ in each round, she gains a non-negative utility. Thus, she also gain a non-negative utility through whole $DCAF$ process. 
    \\
    \textbf{Then we consider WBB.}
    
    According to step 4 and step 5 of $DCAF$, we consider each bidder $i$ who not take part in $DRP$. Her allocation is definitely $\pi_i=\emptyset$ and her payment $p_i=0$.
    For each candidate distributor $i\in T_c$, we consider the payments of $\{p_j(\theta^\prime)\}_{j\in D_{i}(\theta^\prime)}$. According to step 2 of $DRP$, when $\mathcal{\tilde{R}}(\theta^\prime)\geq Rev(\theta^\prime, b_i^*)$, we can get $\sum_{j\in D_{i}(\theta^\prime)}{p_j(\theta^\prime)}=Pr(\theta^\prime, b_i^*)-Rev(\theta^\prime, b_i^*)+\mathcal{\tilde{R}}(\theta^\prime)>0$. According to step 2 of $DRP$, when $\mathcal{\tilde{R}}(\theta^\prime)<Rev(\theta^\prime, b_i^*)$, we can get $\sum_{j\in D_{i}(\theta^\prime)}{p_j(\theta^\prime)}=Pr(\theta^\prime, b_i^\circ)>0$.
    \\
    \textbf{For IC, first we show bidders will report their neighbors truthfully for each round of step 2 to step 7.}
    
    \textbf{Case 1: $i\in T_n$}.
    Candidate distributor consistency $CDP$ ensures that once $i\in T_n$, no matter she reports, bidder $i$ always stay in $T_n$. She will get nothing and pay 0. 
    In this case, bidder $i$ always get nothing and pay 0.
    
    \textbf{Case 2: $i\in T_c$}.
    According to the properties of candidate distributor consistency $CDP$, bidder $i$ can only change herself from $T_c$ to $T_n$. When $i\in T_n$, her utility $u_i(\theta)$ is always 0.
    For $BDP$ satisfies Resale diffusion monotonicity, no matter $r_i^\prime$ reported, bundle tuple $b_i$ will not change.
    Since $\tilde{\mathcal{M}}$ is IC and revenue consistent, bidder $i$ will report truthfully.
    
    \textbf{Case 3: $i\in D_j(\theta^\prime)$, $j\in T_c$}.
    The only part in $DCAF$ can increase $u_i(\theta)$ is $DRP$. Since $\tilde{\mathcal{M}}$ is IC and revenue consistent, bidder $i$ will report truthfully. 
    
    \textbf{Case 4: Otherwise}.
    In this case, bidder $i$ always get nothing and pay 0. 

    \textbf{Then we show bidders will report their valuations truthfully  for each round of step 2 to step 7.}
    
    \textbf{Case 1: $i\in T_n$}.
    Candidate distributor consistency $CDP$ ensures that once $i\in T_n$, no matter she reports, bidder $i$ always stay in $T_n$. She will get nothing and pay 0. 
    In this case, bidder $i$ always get nothing and pay 0.
    
    \textbf{Case 2: $i\in T_c$}.
    According to the properties of candidate distributor consistency $CDP$, bidder $i$ can not change herself from $T_c$ to $T_n$ via misreporting $v_i^\prime$.
    For $BDP$ satisfies Resale diffusion monotonicity, no matter $r_i^\prime$ reported, bundle tuple $b_i$ will not change.
    For $DRP$, $p_i(\theta)$ is independent $v_i(\theta^\prime)$.

    \textbf{Case 3: $i\in D_j(\theta^\prime)$, $j\in T_c$}.
    The only part in $DCAF$ can increase $u_i(\theta)$ is $DRP$. Since $\tilde{\mathcal{M}}$ is IC and revenue consistent, bidder $i$ will report truthfully.
    
    \textbf{Case 4: Otherwise}.
    In this case, bidder $i$ always get nothing and pay 0.

    Thus, the corresponding mechanisms generated by applying $DCAF$ are IC, IR, and WBB.
\end{proof}
\section{Dealer Retail Mechanism}

In this section, we will illustrate a concrete mechanism based on the framework called dealer retail mechanism (DRM). We give the instances of $CDP$, $BDP$, $DRP$ to show the $DRM$. Firstly, we propose a novel bundle division process called Greedy BDP inspired by greedy approach. Then, we find that $DRP$ can use the classical single-item diffusion mechanism IDM, denoted as IDM $DRP$. Finally, we compare the properties of $DRP$ with previous works and show $DRP$ can achieve considerable theory performance.
\subsection{Greedy BDP}
For the bundle division process, we show a simple example using the random approach in \ref{BDP:random}. In addition to the random method, a deterministic method can be given, namely the greedy BDP, which selects the most valuable bundles for each bidder. We use the pricing function which outputs the second price of bidders $T_n$ and resale revenue function which outputs the first price of bidders $T_n$.

\begin{framed}

{\noindent\textbf{Greedy BDP:}} \\
 \noindent\rule{\textwidth}{0.5pt}
 \textbf{Input:} $\theta^\prime$, $\hat{M}$, $T_c$, $T_n$, $Pr$, $Rev$.

\begin{enumerate}
    \item Order the bidders in $T_c$ as $(1,2,\cdots,|T_c|)$.  The ordering rule is that the index of each bidder $i\in T_c$ is independent of report profile $\theta^{\prime}$.   
 \item Sequentially for $1,2,\cdots,|T_c|$: \\
 $b_i^*=\arg\max_{b\subseteq \hat{M}}(\max(v_i(b),Rev(\theta^\prime_{T_n},b))\\  
 -Pr(\theta^\prime_{T_n},b))$ \emph{($b_i^*$ may be $\emptyset$)}, \\ $b_i^\circ=\arg\max_{b\subseteq \hat{M}}(v_i(b)-Pr(\theta^\prime_{T_n},b))$.
 \item Update $\hat{M}=\hat{M}\setminus \{b_i^* \cup b_i^\circ \}$. 
\end{enumerate}
\textbf{Output:} 
The bundle tuple sequence $(b_1,b_2,...,b_{|T_c|})$ divided for each bidder $i\in T_c$.
\end{framed}
\begin{proposition}
    Greedy BDP satisfies resale diffusion monotonicity. 
\end{proposition}

\begin{proof}
    Assume candidate distributor $i\in T_c$ misreported $r_i^{\prime1}$ as $r_i^{\prime2}$ and $r_i^{\prime2}\subseteq r_i^{\prime1}$. According to Greedy BDP, her rank in $T_c$ is independent of $r_i^\prime$. Thus, misreported $r_i^{\prime1}$ as $r_i^{\prime2}$ will not influence utility of bidder $i$. Thus, greedy BDP satisfies resale diffusion monotonicity.
\end{proof}

\subsection{IDM DRP}
We will show that $DRP$ can use the classical single-item diffusion mechanism IDM \cite{li2017mechanism}. We give the formal description of IDM as follows:
\begin{framed}
{\noindent\textbf{Information Diffusion Mechanism (IDM):}} \\
  \noindent\rule{\textwidth}{0.5pt}
\textbf{Input:} $\theta^\prime$
\begin{enumerate}
 \item Denote $m$ as the buyer with the highest valuation report for all item bundle $b=M$, break tie arbitrarily.
 \item Compute the critical diffusion sequence of $m$. Let $C_m^*=\{1,2,...,m\}$ be $m$'s critical diffusion sequence.
 \item For each bidder $i\in C_m$, compute $v_i^*(b)=\max_{j\in D_1(\theta^\prime)\setminus D_i(\theta^\prime)}(v_j(b))$.
 \item The allocation function $\pi$ is:
\begin{equation*}
    \pi_i(\theta^\prime)=
    \begin{cases}
    b, &{i\in C_m\setminus \{m\} \text{ and } v_i^\prime(b)=v_{i+1}^*(b)}\\
    b, &{i=m} \\
    \emptyset, & \text{otherwise} 
    \end{cases}
\end{equation*}
If there exist multiple bidders $i$ with $\pi_i(\theta^\prime)=b$, allocate
the item to the buyer with minimum index $i$ in $C_m$.
\item Assume bidder $w$ win the items under $\pi_i$, the payment function $p$ is:
\begin{equation*}
    p_i(\theta^\prime)=
    \begin{cases}
    v_i^*(b)-v_{i+1}^*(b), &{i\in C_w\setminus \{w\}} \\
    v_i^*(b), &{i=w} \\
    0, & \text{otherwise} 
    \end{cases}
\end{equation*}
\end{enumerate}
\textbf{Output:} 
Allocation $\pi$ and payment $p$.
\end{framed}
 \begin{proposition}
     IDM DRP satisfies revenue consistency.
 \end{proposition}

\begin{proof}
    Assume the candidate distributor as $i\in T_c$. In IDM DRP, the revenue of candidate distributor $\mathcal{\tilde{R}}(\theta^\prime)=Pr(\theta^\prime, b_i^*)-Rev(\theta^\prime, b_i^*)$ is independent of $i$'s critical children $j\in D_{\tilde{s}}$'s reports. Thus, no matter $j\in D_{\tilde{s}}$ misreports, the relationship between $\mathcal{\tilde{R}}(\theta^\prime)$ and $Rev(\theta^\prime,b)$ will not change. Thus, IDM DRP satisfies revenue consistency. 
\end{proof}

\subsection{Dealer Retail Mechanism}

Based on $DCAF$, we propose a new mechanism called Dealer Retail Mechanism (DRM). $DRM$ uses graph exploration $CDP$ to choose the candidate distributors and uses greedy $BDP$ to select bundle tuples for each candidate distributors. Then $DRM$ use IDM $DRP$ to check if the bidder $i$ can resell bundle $b$ to gain potential utility, and giving out the allocation and the payments of $i$'s critical children. 

\begin{framed}
{\noindent\textbf{Dealer Retail Mechanism (DRM):}} \\
 \noindent\rule{\textwidth}{0.5pt}
 \textbf{Input:} $\theta^\prime$, $\mathcal{\tilde{M}}$, $Pr$, $Rev$

\begin{enumerate}
\item Initialize the remain item set $\hat{M}=M$, the unconsidered participant set $\hat{N}=N\cup \{s\}$ 

    \item Use profile $\theta^\prime$ as the input of graph exploration $CDP$, get $CDP(r^\prime)=(T_c,T_n)$, where $T_c\cup T_n\subseteq \hat{N}$. 
    \item Use Greedy BDP to find resale bundles $b_i$ for each candidate distributor $i\in T_c$. 
    \item For each candidate distributor $i\in T_c$, calculate $(\pi_j,p_j)=DRP(\theta^\prime,T_c,(b_i)_{i\in T_c},Pr,Rev,\mathcal{\tilde{M}})$ for each bidder $j\in D_i(\theta^\prime)$.
    \item For each bidder $i\in T_n$, her allocation is $\pi_i=\emptyset$ and payment is $p_i=0$.
     \item Update $\hat{N}=\hat{N}\setminus \{\bigcup_{i\in T_c} D_i(\theta^\prime)\cup T_n\}$ and  $\hat{M}=\hat{M}\setminus\{\bigcup_{j\notin \hat{N}} \pi_j(\theta)\}$.
    \item Then we get a new sub-problem with remain item set $\hat{M}$ unconsidered participant set $\hat{N}$ and $(\theta^\prime_i)_{i \in \hat{N}}$. Repeat step 2 to step 7 until $\hat{M}=\emptyset$ or $\hat{N}=s$.
    \item For each bidder $i\in \hat{N}\setminus s$, her allocation is $\pi_i=\emptyset$ and payment is $p_i=0$.
\end{enumerate}
\textbf{Output:} 
Allocation $\pi$ and payment $p$.
\end{framed}

We now provide an example of DRM depicted in Figure 2 to demonstrate how our mechanism works. First, we conduct graph exploration $CDP$, and get $T_c=\{1,2\}$, $T_n=\{3,4,8,11\}$. By the definition of revenue function and price function, we know that $Rev(a)=3,Rev(b)=9,Rev(ab)=11,Pr(a)=3,Pr(b)=7$ and $Pr(ab)=11$. For bidder 1, $D_1(\theta)=\{5,9\}$, $b_1^*\{b\},b_1^\circ=\{\emptyset\}$. Since $v_5(b)=6<7$, item $b$ does not belong to bidder 1 or 5. For bidder 2, $D_2(\theta)=\{7\}$, $b_2^*\{a\},b_1^\circ=\{a\}$, then bidder 7 gets item $a$ and pay 3; bidder 2 gets nothing and receive 3. Then, we go back to graph exploration CDP and get $T_c=\{6\}$, $T_n=\{\emptyset\}$. By the definition of revenue function and price function, we know that $Rev(a)=Pr(a)=0$. So bidder 6 will reserve item $a$ and pay nothing.

\begin{figure}
    \centering
    \includegraphics[width=\linewidth]{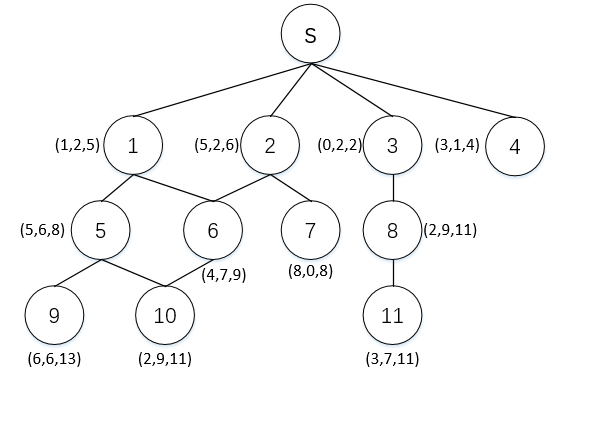}
    \caption{An example of a combinatorial auction on the social network. $Rev(b)=max(v_i(b))_{i\in T_n}$,$Pr(b)$ is the second max $v_i(b)_{i\in T_n}$}
    \label{fig:enter-label}
\end{figure}

Finally, we give the last proposition we need here.

\begin{proposition}
    Graph exploration $CDP$ satisfies candidate distributor
consistency.
\end{proposition}

\begin{proof}
    For every considered bidder $i$ in graph exploration $CDP$, the only influence factor deciding if $i$ belongs $T_c$ is the number of neighbors she reports $|r_i^\prime|$. When other bidders' reports of neighbors not change $r_j^{\prime_1}=r_j^{\prime_2}$ for $j\neq i$, for $i\in T_c$, reports more neighbor will not change $i\in T_c$. When other bidders' reports of neighbors not change, for the bidders $i\in T_n^1$ and $r_i^{\prime1}=r_i$, bidder $i$ still satisfies $i\in T_n$. When other bidders' reports of neighbors not change, for bidder $i\in N\setminus(T_n^1\cup T_c^1)$, change her report of neighbors will not influence $T_c$ and $T_n$.   
\end{proof}

By Proposition 1, Proposition 2, Proposition 3, and Theorem 1, the following is immediately.
\begin{theorem}
   DRM satisfies IC, IR, and WBB.      
\end{theorem}

Then, we propose the definition of existing positive incentive for non-winner which means the mechanism provides positive incentives for some non-winners and prove the mechanisms designed by $DCAF$ satisfy existing positive incentive for non-winner.

\begin{definition}[Existing Positive Incentive for Non-winner]
    A diffusion auction mechanism $\mathcal{M}=(\pi,p)$ is \textbf{existing positive incentive for non-winner (EPI4NW)} if exists a profile $\theta^\prime$ in which exists $i\in N$, $i$'s payment $p_i<0$ when her allocation $\pi_i=\emptyset$.
    
\end{definition}

EPI4NW means that the mechanism $\mathcal{M}$ can bring positive utility to the bidders winning no item. 

\begin{theorem}
   DRM satisfies EPI4NW.      
\end{theorem}
\begin{proof}
    In the description of IDM, we can find that there is a profile $\theta^\prime$ in which exists $i\in C_w\setminus \{w\}$, $i$'s payment $p_i(\theta^\prime)=v_i^*(b)-v_{i+1}^*(b)<0$ although her allocation $\pi_i=\emptyset$. 
    Since DRM use IDM DRP to resale bundles, DRM satisfies EPI4NW.
\end{proof}

Then, we compare the properties of DRP with previous works. In \ref{tab:mechanism}, we show the properties of some previous classic mechanisms e.g. VCG, IDM, MUDAN and so on. We can find that DRM is the only mechanism in combinatorial setting can satisfy IC, IR, WBB, and EPI4NW until now.

\begin{table}[t]
\centering 

\begin{tabular}{|c|c|c|c|c|c|c}
    \hline Mechanism & Scene & IC & IR & WBB & EPI4NW \\ 
    \hline VCG & Single-item & \ding{51} & \ding{51} & \ding{53} & \ding{51} \\ 
    \hline IDM & Single-item & \ding{51} & \ding{51} & \ding{51} & \ding{51} \\ 
    \hline GIDM & Multi-item & \ding{53} & \ding{51} & \ding{51} & \ding{51} \\ 
    \hline DNA & Multi-item & \ding{53} & \ding{51} & \ding{51} & \ding{53} \\ 
    \hline LDM-Tree & Multi-item & \ding{51} & \ding{51} & \ding{51} & \ding{51}\\ 
    \hline MUDAN & Multi-item & \ding{51} & \ding{51} & \ding{51} & \ding{53} \\ 
    \hline MetaMSN-m & Combinatorial & \ding{51} & \ding{51} & \ding{51} & \ding{53} \\ 
    \hline \textcolor{red}{DRM}  &\textcolor{red}{Combinatorial} & \textcolor{red}{\ding{51}} & \textcolor{red}{\ding{51}} & \textcolor{red}{\ding{51}}& \textcolor{red}{\ding{51}} \\ 
    \hline
 
\end{tabular}
\caption{Comparison on diffusion mechanisms over social networks.}
\label{tab:mechanism}
\vspace{-1.0em}
\end{table}

\section{Conclusion}
Currently, there is a significant gap in the design of mechanisms for diffusion auctions involving heterogeneous multi-items. We propose a novel framework, diffusion combinatorial auction framework (DCAF), which leverages single-item diffusion mechanisms to design combinatorial auctions that satisfy incentive compatibility, individual rationality, and weak budget balance conditions. To the best of our knowledge, this work is the first work in the design of diffusion auction mechanisms for  combinatorial setting on social networks. Additionally, we introduce the concept of no positive incentives for non-winners and demonstrate that our proposed mechanisms fulfill this criterion. We can find that DRM is the only mechanism in combinatorial setting can satisfy IC, IR, WBB, and EPI4NW.



\bibliographystyle{ACM-Reference-Format} 
\bibliography{combinatorial}


\end{document}